\documentclass[preprint,3p,12pt,number,sort&compress]{elsarticle}
\usepackage{appendix,amsmath,amsthm,amssymb,bbm,graphicx,hyperref,nccmath,subfigure}

\usepackage{amsmath}
\usepackage{amsthm}
\usepackage{amssymb}
\usepackage{subfigure}
\usepackage{graphicx}
\usepackage{enumitem}
\usepackage[qm]{qcircuit}
\usepackage{amsfonts}
\usepackage{braket}
\usepackage{hyperref}
\usepackage{xcolor}
	\hypersetup
	{
		colorlinks,%
		citecolor=green,%
		linkcolor=blue,%
		urlcolor=blue,%
	}

	\newtheorem{theorem}{Theorem}
	\newtheorem{lemma}[theorem]{Lemma}

\newcommand{\gv}[1]{\ensuremath{\text{\boldmath$ #1 $}}}
\newcommand{\abs}[1]{\left| #1 \right|} 
\newcommand{\norm}[1]{\left\| #1 \right\|} 
\let\baraccent=\= 
\renewcommand{\=}[1]{\stackrel{#1}{=}} 

\newcommand{\aref}[1]{\hyperref[#1]{Appendix~\ref{#1}}}
\newcommand{\lref}[1]{\hyperref[#1]{Lemma~\ref{#1}}}

\begin{document}

\title{Efficient quantum circuits for dense circulant and circulant-like operators}
\author{S. S. Zhou}
\address{Kuang Yaming Honors School, Nanjing University, Nanjing, 210093, China}
\author{J. B. Wang\corref{cor1}}
\cortext[cor1]{jingbo.wang@uwa.edu.au}
\address{School of Physics, The University of Western Australia, Perth WA 6009, Australia}


\begin{abstract}
Circulant matrices are an important family of operators, which have a wide range of applications in science and engineering related fields. They are in general non-sparse and non-unitary.  In this paper, we present efficient quantum circuits to implement circulant operators using fewer resources and with lower complexity than existing methods.  Moreover, our quantum circuits can be readily extended to the implementation of Toeplitz, Hankel, and block circulant matrices.  Efficient quantum algorithms to implement the inverses and products of circulant operators are also provided, and an example application in solving the equation of motion for cyclic systems is discussed. 
\end{abstract}


\maketitle

\section{Introduction}
\label{sec:introduction}

Quantum computation exploits the intrinsic nature of quantum systems in a way that promises to solve problems otherwise intractable on conventional computers.  At the heart of a quantum computer lies a set of qubits whose states are manipulated by a series of quantum logic gates, namely a quantum circuit, to provide the ultimate computational results.  A quantum circuit provides a complete description of a specified quantum algorithm, whose computational complexity is determined by the number of quantum gates required.  However, quantum computation does not always outperform classical computation. In fact there are many known $N$-dimensional matrices that cannot be decomposed as a product of fewer than $N-1$ two-level unitary matrices~\cite{nielsen2010quantum}, and thus cannot be implemented more efficiently on a quantum computer.
An essential research focus in quantum computation is to explore which kinds of linear operations (either unitary or non-unitary) can be efficiently implemented using a series of elementary quantum gates (i.e. two-level unitary matrices) and measurements.

Remarkable progress has been made in such an endeavour, most notably the discovery of Shor's quantum factoring algorithm~\cite{shor1997} and Grover's quantum search algorithms~\cite{grover1996}.
Significant breakthroughs in the area also included the development of efficient quantum algorithms for Hamiltonian simulation, which is central to the studies of chemical and biological processes~\cite{lloyd1996universal,berry2007efficient,childs2010simulating,wiebe2011simulating,poulin2011quantum,berry2015simulating,berry2015hamiltonian}. Recently, Berry, Childs and Kothari presented an algorithm for sparse Hamiltonian simulation achieving near-linear scaling with the sparsity and sublogarithmic scaling with the inverse of the error~\cite{berry2015hamiltonian}. Using the Hamiltonian simulation algorithm as an essential ingredient, Harrow, Hassidim and Lloyd~\cite{harrow2009} showed that for a sparse and well-conditioned matrix $A$, there is an efficient algorithm (known as the HHL algorithm) that provides a quantum state proportional to the solution of the linear system of equations $A\gv{x}=\gv{b}$.

However, as proven by Childs and Kothari~\cite{childs2009limitations}, it is impossible to perform a generic simulation of an arbitrary dense Hamiltonian $H$ in $\mathbb{C}^{N\times N}$ in time $O(poly(\norm{H}, \log N))$, where $\norm{H}$ is the spectral norm, but possible for certain nontrivial classes of Hamiltonians.  It is then natural to ask under what conditions we can extend the sparse Hamiltonian simulation algorithm and the HHL algorithm to the realm of dense matrices. In this paper, we utilise the ``unitary decomposition'' approach developed by Berry, Childs and Kothari~\cite{berry2015simulating} to implement dense circulant Hamiltonians in time $O(poly(\norm{H}, \log N))$. Combining this with the HHL algorithm, we can also efficiently implement the inverse of dense circulant matrices and thus solve systems of circulant matrix linear equations.

Furthermore, we provide an efficient algorithm to implement circulant matrices $C$ directly, by decomposing them into a linear combination of unitary matrices. We then apply the same technique to implement block circulant matrices, Toeplitz and Hankel matrices, which have significant applications in physics, mathematics and engineering~\cite{rietsch2003totally,haupt2010toeplitz,noschese2013tridiagonal,Olson2014,ng2004iterative,peller2012hankel,
kaveh2011block,rjasanow1994effective,tee2005eigenvectors,combescure2009block,petrou2010image}.
For example, we can simulate classical random walks on circulant, Toeplitz and Hankel graphs~\cite{Delanty2012,qiang2016efficient}.
In fact, any arbitrary matrix can be decomposed into a product of Toeplitz matrices~\cite{ye2015every}.  If the number of Toeplitz matrices required is in the order of $O(poly(\log N))$, we can have an efficient quantum circuit.

This paper is organised as follows. In \autoref{sec:circulant}, we present an algorithm to implement circulant matrices, followed by discussions on block circulant matrices, Toeplitz and Hankel matrices in \autoref{sec:beyond}. In \autoref{sec:Hamiltonian} and \autoref{sec:hhl}, we provide a method to simulate circulant Hamiltonians and to implement the inverse of circulant matrices. In \autoref{sec:product}, we describe a technique to efficiently implement products of circulant matrices. In the last section, we provide an example application in solving the equation of motion for vibrating systems with cyclic symmetry. 

\section{Implementation of Circulant Matrices}
\label{sec:circulant}

A circulant matrix has each row right-rotated by one element with respect to the previous row, defined as
\begin{equation}
\label{eq:circulant}
C =
\begin{pmatrix}
c_0 &c_1 &\cdots &\color{black}{c_{N-1}} \\
c_{N-1} &c_0 &\cdots &c_{N-2}\\
\vdots & \vdots & \ddots & \vdots\\
c_1 & c_2 & \cdots & c_0
\end{pmatrix},
\end{equation}
using an $N$-dimensional vector $\gv{c} = (c_0~c_1~\cdots~c_{N-1})$~\cite{golub2012}.  In this paper we will assume $c_j$ to be non-negative for all $j$, which is often the case in practical applications. We also assume that the spectral norm (the largest eigenvalue) $\norm{C}=\sum_{j=0}^{N-1}c_j$ of the circulant matrix $C$ equals to $1$ for simplicity.

Note that $C$ can be decomposed into a linear combination of efficiently realizable unitary matrices as follows,
\begin{equation}
\label{eq:decomposition}
C =
\begin{pmatrix}
c_0 &c_1 &\cdots &c_{N-1} \\
c_{N-1} &c_0 &\cdots &c_{N-2}\\
\vdots & \vdots & \ddots & \vdots\\
c_1 & c_2 & \cdots & c_0
\end{pmatrix}
= c_0
\begin{pmatrix}
1 & 0 &\cdots & 0 \\
0 & 1 &\cdots & 0\\
\vdots & \vdots & \ddots & \vdots\\
0 & 0 & \cdots & 1
\end{pmatrix}
+ c_1
\begin{pmatrix}
0 & 1 &\cdots & 0 \\
0 & 0 &\cdots & \vdots\\
\vdots & \vdots & \ddots & 1\\
1 & 0 & \cdots & 0
\end{pmatrix}
+ \cdots
= \sum_{j=0}^{N-1} c_j V_j,
\end{equation}
where $V_j= \sum_{k=0}^{N-1}\ket{(k-j) \mod N}\bra{k}$.  Such a linear combination of unitary matrices can be dealt with by the unitary decomposition approach introduced by Berry \emph{et al.}~\cite{berry2015simulating}.  For completeness, we restate their method as \lref{thm:sum} given below.

\begin{lemma}
\label{thm:sum}
Let $M = \sum_{\alpha_j}\alpha_j W_j$ be a linear combination of unitaries $W_j$ with $\alpha_j \geq 0$ for all $j$ and $\sum_j \alpha_j = 1$. Let $O_\alpha$ be any operator that satisfies $O_\alpha \ket{0^m} = \sum_j \sqrt{\alpha_j}\ket{j}$, where $m$ is the number of qubits used to represent $\ket{j}$, and $\mathrm{select}(W) = \sum_j\ket{j}\bra{j} \otimes W_j$. Then
\begin{equation}
(O_\alpha^\dagger \otimes I)\mathrm{select}(W) (O_\alpha \otimes I) \ket{0^m}\ket{\psi} = \ket{0^m}M\ket{\psi} + \ket{\Psi^{\perp}},
\end{equation}
where $(\ket{0^m}\bra{0^m} \otimes I) \ket{\Psi^\perp} = 0$.
\end{lemma}

\lref{thm:sum} can be directly applied to implement the circulant matrix $C$, as shown in \autoref{fig:circulant}. Since $\mathrm{select}(V)\ket{j}\ket{k} = \ket{j}\ket{(k-j) \mod N}$, it can be implemented using quantum adders~\cite{draper2000,maynard2013,maynard2013-a,pavlidis2014,cuccaro2004,draper2006}, which requires $O(\log^2N)$ one- or two-qubit gates. We assume for simplicity that $N = 2^{L}$, where $L$ is an integer.

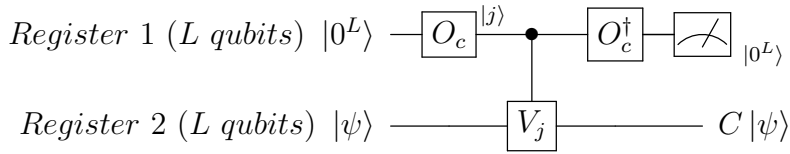
\begin{figure}[ht]
\[
\Qcircuit @C=1em @R=1.3em {
\lstick{Register~1~(L~qubits)~\ket{0^L}}	&	\gate{O_c}_{~~\qquad\qquad\ket{j}}	&	\ctrl{1}	&	\gate{O_c^\dagger}	&	\meter\qw^{~~\qquad \qquad \qquad \ket{0^L}}	\\
\lstick{Register~2~(L~qubits)~\ket{\psi}}	&	\qw	&	\gate{V_j}	&	\qw	&	\rstick{C\ket{\psi}}\qw	\\
}
\]
  \caption{Quantum circuit to implement a circulant matrix.}
  \label{fig:circulant}
\end{figure}

A measurement result of $\ket{0^L}$ in the first register generates the required state $C\ket{\psi}$ in the second register. The probability of this measurement outcome is $O(\norm{C\ket{\psi}}^2)$.
With the help of amplitude amplification~\cite{brassard2002quantum} this can be further improved, requiring only $O(1/\norm{C\ket{\psi}})$ rounds of application of $(O_c^\dagger \otimes I) \mathrm{select}(V) (O_c \otimes I)$.  The amplitude amplification procedure also requires the same number of applications of $O_{\psi}$,  where $O_{\psi}\ket{0^L} = \ket{\psi}$, and its inverse in order to reflect quantum states about the initial state $\ket{0^L}\ket{\psi}$.  If $O_\psi$ is unknown, amplitude amplification is not applicable and we will need to repeat the measuring process in \autoref{fig:circulant} $O(1/\norm{C\ket{\psi}}^2)$ times, during which $O(1/\norm{C\ket{\psi}}^2)$ copies of $\ket{\psi}$ are required.
It is worth noting that with the assumption $c_j \geq 0$, $C$ is unitary if and only if $C=V_j$.  In other words, a non-trivial circulant matrix is non-unitary and therefore, the oblivious amplitude amplification procedure~\cite{berry2014exponential} cannot be applied.

Provided with the oracle $O_c$ satisfying $O_c\ket{0^L} = \sum_{j=0}^{N-1} \sqrt{c_j}\ket{j}$, \autoref{thm:circulant} follows directly from the above discussions. $O_c$ can be efficiently implemented for certain efficiently-computable vectors $\gv{c}$~\cite{grover2002,kaye2004,soklakov2006}. Another way to construct states like $\sum_{j=0}^{N-1} \sqrt{c_j}\ket{j}$ is via qRAM, which uses $O(N)$ hardware resources but only $O(\log N)$ operations to access them~\cite{giovannetti2008,lloyd2013}.

\begin{theorem}[Implementation of Circulant Matrices]
\label{thm:circulant}
There exists an algorithm creating the quantum state $C\ket{\psi}$ for an arbitrary quantum state $\ket{\psi} = \sum_{k=0}^{N-1} \psi_k \ket{k}$, using $O(1/\norm{C\ket{\psi}})$ calls of $O_c$, $O_\psi$ and their inverses, as well as $O(\log^2 N/\norm{C\ket{\psi}})$ additional one- or two-qubit gates.
\end{theorem}

The complexity in \autoref{thm:circulant} is inversely proportional to $p=\norm{C\ket{\psi}}^2$, which depends on the quantum state to be acted upon. Specifically, $\norm{C\ket{\psi}}^2 = \bra{\psi}C^\dagger C\ket{\psi} = \bra{\psi}F \Lambda^\dagger F^\dagger F \Lambda F^\dagger\ket{\psi} = \bra{\psi}F \Lambda^\dagger\Lambda F^\dagger\ket{\psi}$.
Here we use the diagonalization form of $C$~\cite{golub2012}, $C = F \Lambda F^\dagger$, where $F$ is the Fourier matrix with $F_{kj} = e^{2\pi ijk/N}/\sqrt{N}$ and $\Lambda$ is a diagonal matrix of eigenvalues given by $\Lambda_{k} = \sum_{j=0}^{N-1}c_j e^{2\pi ijk/N}$.
Since the spectral norm $\norm{C}$ of the circulant matrix $C$ equals to one, we have $p=\bra{\psi}F \Lambda^\dagger\Lambda F^\dagger\ket{\psi} \geq 1/\kappa^2$, where $\kappa$ is the condition number, defined as the ratio between $C$'s largest and smallest (absolute value of) eigenvalues~\cite{harrow2009}. Therefore, our algorithm is bound to perform well when $\kappa = O(poly(\log N))$. In the ideal case where $\kappa=1$ and $p=1$, the vector $\gv{c}$ is a unit basis in which only one element equals to one and the others are zero.
Even when $\kappa$ is large, our algorithm is efficient when the input quantum state after Fourier Transform is in the subspace whose corresponding eigenvalues are large. To take an extreme but illustrative example, when $\Lambda_k = 1,\; k\neq N/2$ and $\Lambda_{N/2}=0$, we have $\kappa \rightarrow \infty$, but $p = \bra{\phi}\Lambda^\dagger\Lambda\ket{\phi} = \sum_{k=0}^{N-1}\abs{\Lambda_k}^2\abs{\phi_k}^2 \geq 1-\abs{\phi_{N/2}}^2$ in which $\ket{\phi} := F^\dagger \ket{\psi} = \sum_{k=0}^{N-1}\phi_k \ket{k}$. The success rate is therefore lower-bounded by $1-\abs{\phi_{N/2}}^2$, normally close to one.

\section{Circulant-like Matrices}
\label{sec:beyond}

\subsection{Block Circulant Matrices}
\label{subsec:block}

Some block circulant matrices with special structures can also be implemented efficiently in a similar fashion. We assume the blocks are $N'$-dimensional matrices and $L' = \log N'$ in the following discussions.

Firstly, when each block is a unitary operator up to a constant factor (i.e. $\mathbf{C}_j = c_j \mathbf{U}_j$), we have a unitary block (UB) matrix,
\begin{equation}
\label{eq:block}
\begin{split}
C_{UB} =
\begin{pmatrix}
\mathbf{C}_0 &\mathbf{C}_1 &\cdots &\mathbf{C}_{N-1} \\
\mathbf{C}_{N-1} &\mathbf{C}_0 &\cdots &\mathbf{C}_{N-2}\\
\vdots & \vdots & \ddots & \vdots\\
\mathbf{C}_1 & \mathbf{C}_2 & \cdots & \mathbf{C}_0
\end{pmatrix}
&=
\begin{pmatrix}
1 & 0 &\cdots & 0 \\
0 & 1 &\cdots & 0\\
\vdots & \vdots & \ddots & \vdots\\
0 & 0 & \cdots & 1
\end{pmatrix} \otimes \mathbf{C}_0
+
\begin{pmatrix}
0 & 1 &\cdots & 0 \\
0 & 0 &\cdots & \vdots\\
\vdots & \vdots & \ddots & 1\\
1 & 0 & \cdots & 0
\end{pmatrix} \otimes \mathbf{C}_1
+ \cdots\\
&= \sum_{j=0}^{N-1} V_j\otimes\mathbf{C}_j =\sum_{j=0}^{N-1} c_j V_j\otimes\mathbf{U}_j.
\end{split}
\end{equation}
If the set of blocks $\{\mathbf{U}_j\}_{j=0}^{N-1}$
can be efficiently implemented, then by simply replacing $\mathrm{select}(V) = \sum_{j=0}^{N-1}\ket{j}\bra{j}\otimes V_j$ with $\sum_{j=0}^{N-1}\ket{j}\bra{j}(V_j\otimes\mathbf{U}_j)$, we can efficiently implement the block circulant matrices $C_{UB}$ using the same algorithm discussed in \autoref{sec:circulant} as illustrated in \autoref{subfig:block}.

Specifically, when the set of blocks $\{\mathbf{U}_j\}_{j=0}^{N-1}$ are one-dimensional, we can implement complex-valued circulant matrices with efficiently computable phase. For example, for $\mathbf{U}_j = (e^{i \theta j}), ~j=0,1,\ldots,N-1$, circulant matrices with the parameter vector $\gv{c} = (c_0, e^{i\theta}c_1, \ldots,$\\$e^{i(N-1)\theta}c_{N-1}) $ can be implemented efficiently. Moreover, if $\theta = \pi$, $\gv{c} = (c_0, -c_1, \ldots,(-1)^{N-1}$\\$c_{N-1})$ corresponding to the circulant matrix with negative elements are on odd-numbered sites is efficiently-implementable.

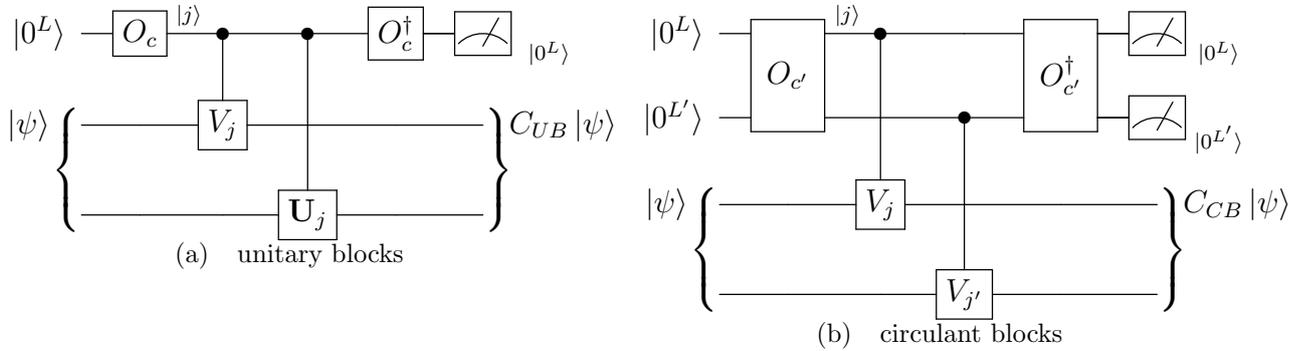
\begin{figure}[ht]
    \centering
    \subfigure[~~unitary blocks]
    { \label{subfig:block} 
    \Qcircuit @C=1em @R=1.3em {
\lstick{\ket{0^{L}}}	&	\gate{O_c}_{\quad\qquad\qquad\ket{j}}	&	\ctrl{1}	&	\ctrl{2}	&	\gate{O_c^\dagger}	&	\meter\qw^{\quad\qquad \qquad \qquad \ket{0^L}}	\\
\lstick{\ket{\psi}~}	&	\qw	&	\gate{V_j}	&	\qw	&	\qw	&	\rstick{~C_{UB}\ket{\psi}}\qw	\\
	&	\qw	&	\qw	&	\gate{\mathbf{U}_j}	&	\qw	&	\qw\gategroup{2}{1}{3}{1}{1em}{\{}\gategroup{2}{6}{3}{6}{1em}{\}}	\\
    }
    }
    $\quad$
    \subfigure[~~circulant blocks]
    { \label{subfig:BCCB} 
    \Qcircuit @C=1em @R=1.3em {
\lstick{\ket{0^{L}}}	&	\multigate{1}{O_{c'}}_{\qquad\qquad\qquad\ket{j}}	&	\ctrl{2}	&	\qw	&	\multigate{1}{O_{c'}^\dagger}	&	\meter\qw^{\quad\qquad \qquad \qquad \ket{0^L}}	\\				
\lstick{\ket{0^{L'}}}	&	\ghost{O_{c'}}	&	\qw	&	\ctrl{2}	&	\ghost{O_{c'}^\dagger}	&	\meter\qw^{\quad\qquad \qquad \qquad \ket{0^{L'}}}	\\				
\lstick{\ket{\psi}~}	&	\qw	&	\gate{V_j}	&	\qw	&	\qw	&	\rstick{~C_{CB}\ket{\psi}}\qw	\\				
	&	\qw	&	\qw	&	\gate{V_{j'}}	&	\qw	&	\qw\gategroup{3}{1}{4}{1}{1em}{\{}\gategroup{3}{6}{4}{6}{1em}{\}}	\\				
    }
    }
    \caption{The quantum circuit to implement block circulant matrices with special structures.}
\end{figure}

Another important family is block circulant matrices with circulant blocks (CB), which has found a wide range of applications in algorithms, mathematics, etc.~\cite{rjasanow1994effective,tee2005eigenvectors,combescure2009block,petrou2010image}. It is defined as follows
\begin{equation}
\label{eq:BCCB}
\begin{split}
C_{CB} =
\begin{pmatrix}
\mathbf{C}_0 &\mathbf{C}_1 &\cdots &\mathbf{C}_{N-1} \\
\mathbf{C}_{N-1} &\mathbf{C}_0 &\cdots &\mathbf{C}_{N-2}\\
\vdots & \vdots & \ddots & \vdots\\
\mathbf{C}_1 & \mathbf{C}_2 & \cdots & \mathbf{C}_0
\end{pmatrix},
\end{split}
\end{equation}
where $\mathbf{C}_j$ is a circulant matrix specified by a $N'$-dimensional vector $\gv{c}_j = (c_{j0}~c_{j1}~\cdots~c_{j(N'-1)})$. $C_{CB}$ is a $N\times N'$-dimensional matrix determined by $N\times N'$ parameters $\{c_{jj'}\}_{\substack{j=0,\ldots,N-1 \\ j'=0,\ldots,N'-1}}$. It can be decomposed as follows
\begin{equation}
\label{eq:BCCBdecomposition}
\begin{split}
C_{CB} = \sum_{j=0}^{N-1} \sum_{j'=0}^{N'-1} c_{jj'} V_{j} \otimes V_{j'}.
\end{split}
\end{equation}
Given an oracle $O_{c'}$ satisfying $O_{c'}\ket{0^{L+L'}} = \sum_{j=0}^{N-1}\sum_{j=0}^{N'-1} c_{jj'} \ket{j}\ket{j'}$, we can implement $C_{CB}$ using the quantum circuit shown in \autoref{subfig:BCCB}, which adopts a combination of two quantum subtractors.

\subsection{Toeplitz and Hankel Matrices}
\label{subsec:ToeplitzHankel}

A Toeplitz matrix is a matrix in which each descending diagonal from left to right is constant, which can be written explicitly as
\begin{equation}
T =
\begin{pmatrix}
t_{0} & t_{-1} & \cdots & t_{-(N-2)} & t_{-(N-1)} \\
t_1   & t_{0}  & \cdots & t_{-(N-3)} & t_{-(N-2)} \\
t_2   & t_1    & \cdots & t_{-(N-4)} & t_{-(N-3)} \\
\vdots &\vdots & \ddots &\vdots      &\vdots      \\
t_{N-1}&t_{N-2}& \cdots & t_1        & t_0        \\
\end{pmatrix},
\end{equation}
specified by $2N-1$ parameters. We focus on the situation where $t_j \geq 0$ for all $j$ as in \autoref{sec:circulant}. Clearly, when $t_{-(N-i)}=t_i$ for all $i$, $T$ is a circulant matrix. Although a Toeplitz matrix is not circulant in general, any Toeplitz matrix $T$ can be embedded in a circulant matrix~\cite{ng2004iterative,mahasinghe2016efficient}, defined by
\begin{equation}
C_T =
\begin{pmatrix}
T &B_T \\
B_T& T \\
\end{pmatrix},
\end{equation}
where $B_T$ is another Toeplitz matrix defined by
\begin{equation}
B_T =
\begin{pmatrix}
0 & t_{N-1} & \cdots & t_{2} & t_{1} \\
t_{-(N-1)}   & 0  & \cdots & t_{3} & t_{2} \\
t_{-(N-2)}   & t_{-(N-1)}    & \cdots & t_{4} & t_{3} \\
\vdots &\vdots & \ddots &\vdots      &\vdots      \\
t_{-1}&t_{-2}& \cdots & t_{-(N-1)}        & 0        \\
\end{pmatrix}.
\end{equation}
As a result, we use this embedding to implement Toeplitz matrices because
\begin{equation}
\begin{pmatrix}
T & B_T \\
B_T & T \\
\end{pmatrix}
\begin{pmatrix}
\psi\\
0\\
\end{pmatrix}
=
\begin{pmatrix}
T \psi\\
B_T \psi \\
\end{pmatrix} .
\end{equation}

Therefore, by implementing $C_T$, we obtain a quantum state proportional to $\ket{0}T\ket{\psi}+\ket{1}B_T\ket{\psi}$. Then we do a quantum measurement on the single qubit (in the second register in \autoref{fig:Toeplitz}) to obtain the quantum state $T\ket{\psi}$. The success rate is $\norm{T\ket{\psi}}^2$ according to \autoref{thm:circulant} under the normalization condition that $\sum_{j=-(N-1)}^{N-1}t_j = \sum_{j=0}^{N-1} c_j = 1$. With the help of amplitude amplification, only $O(1/\norm{T\ket{\psi}})$ applications of the circuit in \autoref{fig:Toeplitz} are required.

\begin{figure}[ht]
\[
\Qcircuit @C=1em @R=1.3em {
\lstick{Register~1~(L+1~qubits)~\ket{0^{L+1}}}	&	\gate{O_c}_{~~\qquad\qquad\ket{j}}	&	\ctrl{1}	&	\gate{O_c^\dagger}	&	\meter\qw^{~~\qquad \qquad \qquad \ket{0^L}}	\\
\lstick{Register~2~(1~qubit)~\ket{0}}	&	\qw	&	\multigate{1}{V_j}	&	\qw	&	\meter\qw^{\qquad \qquad \qquad \ket{0}}	\\
\lstick{Register~3~(L~qubits)~\ket{\psi}}	&	\qw	&	\ghost{V_j}	&	\qw	&	\rstick{T\ket{\psi}}\qw	\\
}
\]
  \caption{The quantum circuit to implement a Toeplitz matrix. In this figure, $O_c\ket{0^{L+1}} = \sum_{j=0}^{2N-1}c_j\ket{j}$ where $\gv{c} = (t_0~t_{-1}\cdots t_{-(N-1)}~0~t_{N-1}\cdots t_1) $.}
  \label{fig:Toeplitz}
\end{figure}
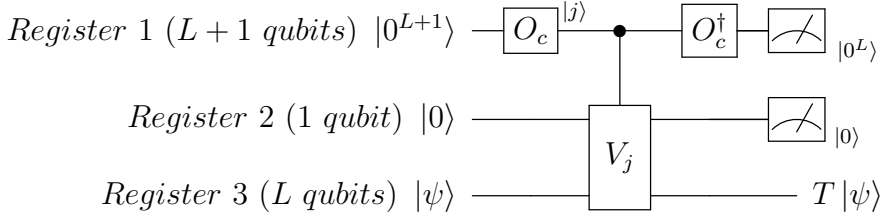

A Hankel matrix is a square matrix in which each ascending skew-diagonal from left to right is constant, which can be written explicitly as
\begin{equation}
H =
\begin{pmatrix}
h_{N-1} & h_{N-2} & \cdots & h_{1} & h_{0} \\
h_{N-2} & h_{N-3}  & \cdots & h_{0} & h_{-1} \\
\vdots &\vdots & \ddots &\vdots      &\vdots      \\
h_{1}  & h_{0}    & \cdots & h_{-(N-3)} & h_{-(N-2)} \\
h_{0}  &h_{-1}& \cdots & h_{-(N-2)}        & h_{-(N-1)}        \\
\end{pmatrix}.
\end{equation}
specified by $2N-1$ non-negative parameters. A permutation matrix $P = \sigma_x^{\otimes L}$ transforms a Hankel matrix into a Toeplitz matrix. It can be easily verified that $T=HP$ and $H=TP$, in which $t_j = h_j$ for all $j$.

Therefore by inserting the permutation $P$ before the implementation of $T$, the circuit in \autoref{fig:Toeplitz} can be used to implement $H$, and the success rate is $\norm{H\ket{\psi}}^2$ under the normalization condition that $ \sum_{j=-(N-1)}^{N-1}h_j = \sum_{j=0}^{N-1} c_j = 1$. With the help of amplitude amplification, only $O(1/\norm{H\ket{\psi}})$ applications are required.

In comparison with existing algorithms, such as that described in \cite{mahasinghe2016efficient}, 
the above described quantum circuit provides a better way to realize circulant-like matrices, requiring fewer resources and with lower complexity. For example, only $2 \log N$ qubits are required to implement $N$-dimensional Toeplitz matrices, which is a significant improvement over the algorithm presented in \cite{mahasinghe2016efficient} via sparse Hamiltonian simulations. More importantly, this is an exact method and its complexity does not depend on an error term. It is also not limited to sparse circulant matrices $C$ as in \cite{mahasinghe2016efficient}.  Moreover, implementation of non-unitary matrices, such as circulant matrices, is not only of importance in quantum computing, but also a significant ingredient in quantum channel simulators~\cite{lu2015universal,wang2015quantum}, because the set of Kraus operators in the quantum channel $\rho \mapsto \sum_{i} K_i \rho K_i^\dagger$ is normally non-unitary~\cite{nielsen2010quantum}. The simplicity of our circuit increases its feasibility in experimental realizations.

\section{Circulant Hamiltonians}
\label{sec:Hamiltonian}

Hamiltonian simulation is expected to be one of the most important undertakings for quantum computation.  It is therefore important to explore the possibility of efficient implementation of circulant Hamiltonians due to their extensive applications. Particularly, the implementation of $e^{-iCt}$ is equivalent to the implementation of continuous-time quantum walks on a weighted circulant graph~\cite{wang2013physical,childs2010relationship}. Moreover, simulation of Hamiltonians is also an important part in the HHL algorithm to solve linear systems of equations~\cite{harrow2009}.

A number of algorithms have been shown to be able to efficiently simulate sparse Hamiltonians~\cite{lloyd1996universal,berry2007efficient,childs2010simulating,wiebe2011simulating,poulin2011quantum,berry2015simulating,berry2015hamiltonian}, including the unitary decomposition approach~\cite{berry2015simulating}. We show that this approach can be extended to the simulation of dense circulant Hamiltonians. As we know, circulant matrices are diagonalizable and $e^{-iCt} = F e^{-i\Lambda t} F^\dagger$. Hence, there is a direct method to implement $e^{-iCt}$~\cite{qiang2016efficient} when its diagonal elements $\{\Lambda_k\}_{k=0}^{N-1}$ are already known. However, this method is generally not extensible when $\{c_j\}_{j=0}^{N-1}$ are inputs.


In this section, we will focus on the simulation of Hermitian circulant matrices, when $e^{-iCt}$ is unitary. For completeness, we first describe briefly the unitary decomposition approach and then discuss how it can be used to efficiently simulate dense circulant Hamiltonians. To simulate $U = e^{-iCt}$, we divide the evolution time $t$ into $r$ segments with $U_r = e^{-iCt/r}$, which can be approximated as $\tilde{U} = \sum_{k=0}^{K}1/k!(-iCt/r)^k$ with error $\epsilon$. It can be proven that if we choose $K = O\Big(\frac{\log(r/\epsilon)  }{\log\log(r/\epsilon) }\Big) = O\Big(\frac{\log(t/\epsilon)  }{\log\log(t/\epsilon) }\Big)$, then $\|U_r - \tilde{U}\| \leq \epsilon/r$ and the total error is within $\epsilon$.

Since $C = \sum_{j=0}^{N-1} c_jV_j$ as given by \autoref{eq:decomposition}, we have
\begin{equation}
\tilde{U} = \sum_{k=0}^{K}\frac{(-iCt/r)^k}{k!} = \sum_{k=0}^K \sum_{j_1,\ldots,j_k=0}^{N-1} \frac{(-it/r)^k}{k!}
c_{j_1}\cdots c_{j_k} V_{j_1}\cdots V_{j_k}.
\end{equation}
According to \lref{thm:sum}, let $W_{(k,j_1,\ldots,j_k)}=(-i)^kV_{j_1}\cdots V_{j_k}$ and
\begin{equation}
\label{eq:Oalpha}
O_\alpha \ket{0^{K+KL}} = \frac{1}{\sqrt{s}}\sum_{k=0}^{K} \sum_{j_1,\ldots,j_k=0}^{N-1} \sqrt{(t/r)^k/k! c_{j_1}\cdots c_{j_k}} \ket{1^{k}0^{K-k}}\ket{j_1}\cdots\ket{j_k}\ket{0^{(K-k)L}},
\end{equation}
where $\ket{1^{k}0^{K-k}}$ is the unary encoding of $k$. Here $s$ is the normalization coefficient and we choose $r = \lceil t/\ln 2 \rceil$ so that
\begin{equation}
s = \sum_{k=0}^{K}\sum_{j_1,\ldots,j_k=0}^{N-1} \frac{(t/r)^k}{k!} c_{j_1}\cdots c_{j_k}
= \sum_{k=0}^{K} \frac{\big((c_0+\cdots+c_{N-1})t/r\big)^k}{k!}  \approx 2.
\end{equation}
Then we have
\begin{equation}
(O_\alpha^\dagger \otimes I) \mathrm{select}(W) (O_\alpha \otimes I) \ket{0^{K+KL}}\ket{\psi} = \frac{1}{s} \ket{0^{K+KL}} \tilde{U} \ket{\psi} + \ket{\Psi^\perp},
\end{equation}
where $(\ket{0^{K+KL}}\bra{0^{K+KL}} \otimes I) \ket{\Psi^\perp} = 0$. It has been shown in Ref.~\cite{berry2015simulating} that after one step of oblivious amplitude amplification procedure~\cite{berry2014exponential}, $U_r = e^{-iCt/r}$ can be simulated within error $\epsilon/r$. The oblivious amplitude amplification procedure avoids the repeated preparations of $\ket{\psi}$ so that $\tilde{U}\ket{\psi}$ can be obtained using only one copy of $\ket{\psi}$. The total complexity depends on the number of gates required to implement $\mathrm{select}(W)$ and $O_\alpha$.

\begin{figure}
{\tiny
\[
\Qcircuit @C=1em @R=1.3em{
\lstick{\ket{0}_1}	&	\multigate{3}{R_{ini}}	&	\ctrl{4}	&	\qw	&	\qw	&	\qw	&	\qw	&	\qw	&	\qw	&	\qw	&	\qw	&	\qw	&	\qw	&	\qw	&	\gate{-i}	&	\qw	&	\qw	&	\qw	&	\qw	&	\qw	&	\ctrl{4}	&	\multigate{3}{R^\dagger_{ini}}	&	\meter\qw^{\quad\qquad \qquad \qquad \ket{0}}	\\
\lstick{\ket{0}_2}	&	\ghost{R_{ini}}	&	\qw	&	\ctrl{4}	&	\qw	&	\qw	&	\qw	&	\qw	&	\qw	&	\qw	&	\qw	&	\qw	&	\qw	&	\qw	&	\gate{-i}	&	\qw	&	\qw	&	\qw	&	\qw	&	\ctrl{4}	&	\qw	&	\ghost{R^\dagger_{ini}}	&	\meter\qw^{\quad\qquad \qquad \qquad \ket{0}}	\\
	&		&		&		&		&		&		&		&		&		&		&		&		&		&	\vdots	&		&		&		&		&		&		&		&		\\
\lstick{\ket{0}_K}	&	\ghost{R_{ini}}	&	\qw	&	\qw	&	\qw	&	\cdots	&	\quad	&	\ctrl{4}	&	\qw	&	\qw	&	\qw	&	\qw	&	\qw	&	\qw	&	\gate{-i}	&	\ctrl{4}	&	\qw	&	\cdots	&	\quad	&	\qw	&	\qw	&	\ghost{R^\dagger_{ini}}	&	\meter\qw^{\quad\qquad \qquad \qquad \ket{0}}	\\
\lstick{\ket{0^L}_1}	&	\qw	&	\gate{O_c}	&	\qw	&	\qw	&	\cdots	&	\quad	&	\qw	&	\ctrl{4}_{\quad\ket{j_1}}	&	\qw	&	\qw	&	\qw	&	\qw	&	\qw	&	\qw	&	\qw	&	\qw	&	\cdots	&	\quad	&	\qw	&	\gate{O^\dagger_c}	&	\qw	&	\meter\qw^{\quad\qquad \qquad \qquad \ket{0^L}}	\\
\lstick{\ket{0^L}_2}	&	\qw	&	\qw	&	\gate{O_c}	&	\qw	&	\cdots	&	\quad	&	\qw	&	\qw	&	\ctrl{3}_{\quad\ket{j_k}}	&	\qw	&	\cdots	&	\quad	&	\qw	&	\qw	&	\qw	&	\qw	&	\cdots	&	\quad	&	\gate{O^\dagger_c}	&	\qw	&	\qw	&	\meter\qw^{\quad\qquad \qquad \qquad \ket{0^L}}	\\
	&		&	\vdots	&		&		&		&		&		&		&		&		&		&		&		&		&		&		&		&		&		&	\vdots	&	\vdots	&		\\
\lstick{\ket{0^L}_K}	&	\qw	&	\qw	&	\qw	&	\qw	&	\cdots	&	\quad	&	\gate{O_c}	&	\qw	&	\qw	&	\qw	&	\cdots	&	\quad	&	\ctrl{1}_{\quad\ket{j_K}}	&	\qw	&	\gate{O^\dagger_c}	&	\qw	&	\cdots	&	\quad	&	\qw	&	\qw	&	\qw	&	\meter\qw^{\quad\qquad \qquad \qquad \ket{0^L}}	\\
\lstick{\ket{\psi}}	&	\qw	&	\qw	&	\qw	&	\qw	&	\qw	&	\qw	&	\qw	&	\gate{V_{j_1}}	&	\gate{V_{j_k}}	&	\qw	&	\cdots	&	\quad	&	\gate{V_{j_K}}	&	\qw	&	\qw	&	\qw	&	\qw	&	\qw	&	\qw	&	\qw	&	\qw\gategroup{1}{1}{9}{8}{1em}{_\}}\gategroup{1}{9}{9}{15}{1em}{_\}}\gategroup{1}{16}{9}{22}{1em}{_\}}	&	\rstick{e^{-iCt/r}\ket{\psi}}\qw	\\
	&		&	\mbox{$O_\alpha$}	&		&		&		&		&		&		&		&	\mbox{$\mathrm{select}(W)$}	&		&		&		&		&		&		&		&		&	\mbox{$O_\alpha^\dagger$}	&		&		&		
}
\]
}
\caption{The quantum circuit to implement one segment of circuilant Hamiltonians. Here $\boxed{R_{ini}}\ket{0^K} = \sum_{k=0}^{K} \sqrt{(t/r)^k/k!} \ket{1^k 0^{K-k}}$ and $\boxed{-i}=\ket{0}\bra{0}+(-i)\ket{1}\bra{1}$.}
\label{fig:Hamiltonian}
\end{figure}
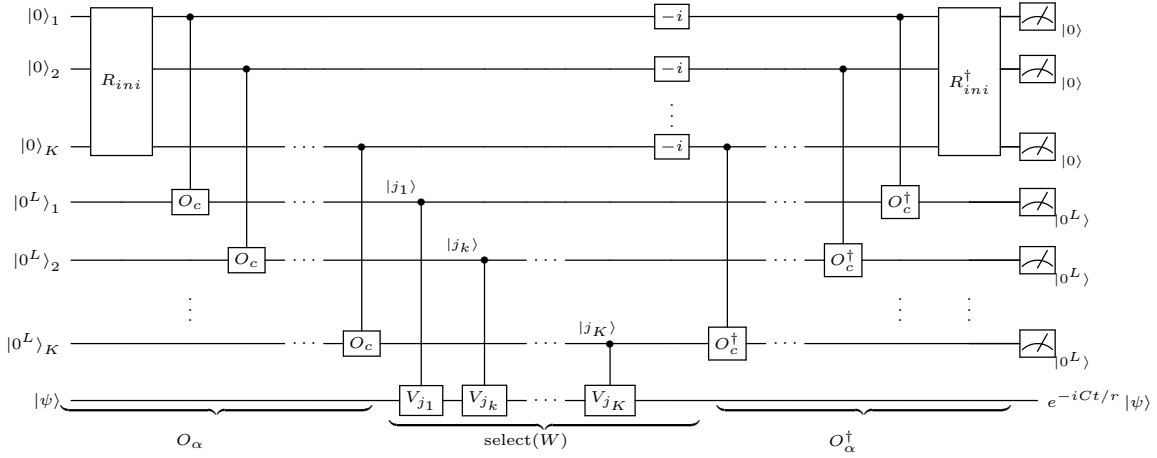

\begin{theorem}[Simulation of Circulant Hamiltonians]
\label{thm:Hamiltonian}
There exists an algorithm performing $e^{-iCt}$ on an arbitrary quantum state $\ket{\psi}$ within error $\epsilon$, using $O\Big(t\frac{\log(t/\epsilon)}{\log\log(t/\epsilon)}\Big)$ calls of controlled-$O_c$~\footnote{By controlled-$O_c$, we mean the operation $\ket{0}\bra{0}\otimes I + \ket{1}\bra{1}\otimes O_c$} and its inverse, as well as $O\Big(t(\log N)^2\frac{\log(t/\epsilon)}{\log\log(t/\epsilon)}\Big)$ additional one- and two-qubit gates.
\end{theorem}

\begin{proof}
We first consider the number of gates used to implement $O_\alpha$ in \autoref{eq:Oalpha}. It can be decomposed into two steps. The first step is to create the normalized version of the state $\sum_{k=0}^{K} \sqrt{(t/r)^k/k!} \ket{1^k 0^{K-k}}$ from the initial state $\ket{0^K}$, which takes $O(K)$ consecutive one-qubit rotations on each qubit. We then apply $K$ sets of controlled-$O_c$ to transform $\ket{0^L}$ into $\sum_{j=0}^{N-1}\sqrt{c_j}\ket{j}$ when the control qubit is $\ket{1}$. We therefore need $O(K)$ calls of controlled-$O_c$ and $O(K)$ additional one-qubit gates to implement $O_\alpha$.

Next we focus on the implementation of
\begin{multline}
\mathrm{select}(W) = \sum_{(k,j_1,j_2,\cdots,j_k)}\\ \ket{1^{k}0^{K-k}}\ket{j_1}\cdots\ket{j_k}\ket{0^{(K-k)L}} \bra{1^{k}0^{K-k}}\bra{j_1}\cdots\bra{j_k}\bra{0^{(K-k)L}}\otimes (-i)^{k}V_{j_1}\cdots V_{j_k},
\end{multline}
which performs the transformation
\begin{multline}
\label{eq:specific}
\ket{1^{k}0^{K-k}}\ket{j_1}\cdots\ket{j_k}\ket{0^{(K-k)L}}\ket{\psi} \xrightarrow{\text{select}(W)}\\
\ket{1^{k}0^{K-k}}\ket{j_1}\cdots\ket{j_k}\ket{0^{(K-k)L}}(-i)^k V_{j_1}\cdots V_{j_k}\ket{\psi}.
\end{multline}
As $V_j\ket{\ell} = \ket{(\ell-j) \mod N}$, we can transform $\ket{\psi}$ into $V_{j_1}\cdots V_{j_k}\ket{\psi}$ by applying $K$ quantum subtractors between $\ket{j}_{j = j_1,j_2,\ldots,j_k,0,\ldots}$ and $\ket{\psi}$. $K$ phase gates on each of the first $K$ qubits multiply the amplitude by $(-i)^k$. Therefore $\mathrm{select}(W)$ can be decomposed into $O(K \log^2 N)$ numbers of one or two-qubit gates.

In summary, $O(K)$ calls of controlled-$O_c$ and its inverse as well as $O(K \log^2N)$ additional one-qubit gates are sufficient to implement one segment $e^{-iCt/r}$. And the total complexity to implement $r$ segments will be
$O(tK)$ calls of controlled-$O_c$ and its inverse as well as $O(tK \log^2N)$ additional one-qubit gates, where $K = O\Big(\frac{\log(t/\epsilon)  }{\log\log(t/\epsilon) }\Big)$.
\end{proof}

Note that we assumed the spectral norm $\norm{C}=1$. To explicitly put it in the complexity in \autoref{thm:Hamiltonian}, we can simply replace $t$ by $\norm{C}t$.

\section{Inverse of Circulant Matrices}
\label{sec:hhl}

We now show that the HHL algorithm can be extended to solve systems of circulant matrix linear equations.

\begin{theorem}[Inverse of Circulant Matrices]
\label{thm:hhl}
There exists an algorithm creating the quantum state $C^{-1}\ket{\psi}/\norm{C^{-1}\ket{\psi}}$ within error $\epsilon$ given an arbitrary quantum state $\ket{\psi}$, using $\tilde{O}(\kappa^2/\epsilon)$~\footnote{We use the symbol $\tilde{O}$ to suppress polylogarithmic factors.}
calls of controlled-$O_c$ and its inverse, $O(\kappa)$ calls of $O_\psi$, as well as $\tilde{O}(\kappa^2\log^2 N/\epsilon)$ additional one- and two-qubit gates.
\end{theorem}
\begin{proof}
The procedure of the HHL algorithm works as follows~\cite{harrow2009}:
\begin{enumerate}
\item Apply the oracle $O_\psi$ to create the input quantum state $\ket{\psi}$:
\[\ket{0^L}\xrightarrow{O_\psi}\ket{\psi} = \sum_{j=0}^{N-1}b_j\ket{u_j},\]
where $\{\ket{u_j}\}_{j=0}^{N-1}$ are the eigenvectors of $C$.
\item Run phase estimation of the unitary operator $e^{i2\pi C}$:
\[\sum_{j=0}^{N-1}b_j\ket{u_j}\rightarrow\sum_{j=0}^{N-1}b_j\ket{u_j}\ket{\Lambda_j},\]
where $\Lambda_j$ are the eigenvalues of $C$ and $\Lambda_j \leq 1$.
\label{item:estimation}
\item Perform a controlled-rotation on an ancillary qubit:
\[\sum_{j=0}^{N-1}b_j\ket{u_j}\ket{\Lambda_j}\rightarrow\sum_{j=0}^{N-1}b_j\ket{u_j}\ket{\Lambda_j}\Big(1/(\kappa\Lambda_j) \ket{1} + \sqrt{1-1/(\kappa^2\Lambda_j^2)}\Big),\]
where $\kappa$ is the condition number defined in \autoref{sec:circulant} to make sure that $1/(\kappa\Lambda_j)\leq 1$ for all $j$.
\item Undo the phase estimation and then measure the ancillary qubit. Conditioned on getting $1$, we have an output state $\propto \sum_{j=0}^{N-1}b_j/\Lambda_j\ket{u_j}$ and the success rate $p = \sum_{j=0}^{N-1}\abs{b_j/\kappa\Lambda_j}^2 = \Omega(1/\kappa^2)$.
\end{enumerate}
Error occurs in Step~\ref{item:estimation} in Hamiltonian simulation and phase estimation. The complexity scales sublogarithmically with the inverse of error in Hamiltonian simulation as in \autoref{thm:Hamiltonian} and scales linearly with it in phase estimation~ \cite{nielsen2010quantum}. The dominant source of error is phase estimation. Following from the error analysis in Ref.~\cite{harrow2009}, a precision $O(\kappa/\epsilon)$ in phase estimation results in a final error $\epsilon$. Taking the success rate $p=\Omega(1/\kappa^2)$ into consideration, the total complexity would be $\tilde{O}(\kappa^2/\epsilon)$, with the help of amplitude amplification~\cite{brassard2002quantum}.
\end{proof}

\section{Products of Circulant Matrices}
\label{sec:product}

Products of circulant matrices are also circulant matrices, because a circulant matrix can be decomposed into a linear combination of $\{V_j\}_{j=0}^{N-1}$ that constitute a cyclic group of order $N$ (we have $V_j V_k = V_{(j+k) \mod N}$).
Suppose $C^{(1,2)} = C^{(1)}C^{(2)}$ is the product of two circulant matrices $C^{(1)}$ and $C^{(2)}$ which have a parameter vector $\gv{c}^{(1,2)}$, where
\begin{equation}
\label{eq:product}
c^{(1,2)}_j = \sum_{\substack{j_1,j_2 \\ j_1+j_2 \equiv j \mod N}} c^{(1)}_{j_1}c^{(2)}_{j_2},
\end{equation}
where $\gv{c}^{(1)}$ and $\gv{c}^{(2)}$ are each the parameters of $C^{(1)}$ and $C^{(2)}$.
Clearly, when the spectral norm of $C^{(1)}$ and $C^{(2)}$ are one, the spectral norm of $C^{(1,2)}$ is also one. Classically, to calculate the parameters $\gv{c}^{(1,2)}$ would take up $O(N)$ space. However, in the quantum case, we will show that $O_{c^{(1,2)}}$, encoding $\gv{c}^{(1,2)}$, can be prepared using one $O_{c^{(1)}}$ and one $O_{c^{(2)}}$. It means that the oracle for a product of circulant matrices can be efficiently prepared when its factor circulants are efficiently implementable, as illustrated in \autoref{fig:product}.

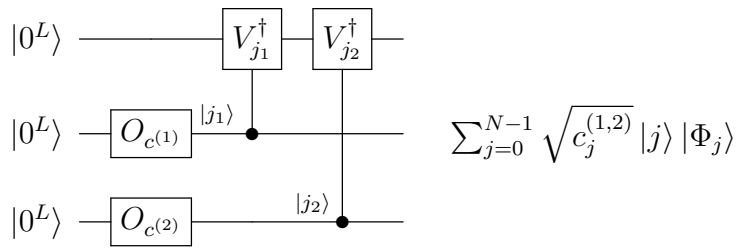
\begin{figure}[ht]
\[
\Qcircuit @C=1em @R=1.3em{
\lstick{\ket{0^L}}	&	\qw	&	\gate{V_{j_1}^\dagger}	&	\gate{V_{j_2}^\dagger}	&	\qw&	\\
\lstick{\ket{0^L}}	&	\gate{O_{c^{(1)}}}	&	\ctrl{-1}_{\quad\ket{j_1}}	&	\qw	&	\qw&\rstick{\sum_{j=0}^{N-1}\sqrt{c_j^{(1,2)}}\ket{j}\ket{\Phi_j}}\\
\lstick{\ket{0^L}}	&	\gate{O_{c^{(2)}}}	&	\qw	&	\ctrl{-2}_{\quad\ket{j_2}}	&	\qw&	\\
}
\]
\caption{The quantum circuit of $O_{c^{(1,2)}}$. Here $V_j^\dagger = \sum_{k=0}^{N-1}\ket{(k+j) \mod N}\bra{k}$ and controlled-$V_j^\dagger$ is a quantum adder.}
\label{fig:product}
\end{figure}

\begin{theorem}[Products of Circulant Matrices]
\label{thm:product}
There exists an algorithm creating the oracle $O_{c^{(1,2)}}$, which satisfies
\begin{equation}
\label{eq:prodoracle}
O_{c^{(1,2)}} \ket{0^{3L}} =  \sum_{j=0}^{N-1} \sqrt{c^{(1,2)}_j} \ket{j}\ket{\Phi_j},
\end{equation}
where $\ket{\Phi_j}$ is a unit quantum state dependent on $j$, using one $O_{c^{(1)}}$, one $O_{c^{(2)}}$ and $O(\log^2 N)$ additional one- and two-qubit gates.
\end{theorem}

\begin{proof}
We need $2L$ ancillary qubits divided into two registers to construct the oracle for the product of two circulant matrices.
We start by applying $O_{c^{(1)}}$ and $O_{c^{(2)}}$ on the last $2$ registers, we obtain
\begin{equation}
\ket{0^{3L}} \rightarrow \ket{0^L}\big(\sum_{j_1=0}^{N-1}\sqrt{c^{(1)}_{j_1}}\ket{j_1}\big)\big(\sum_{j_2=0}^{N-1}\sqrt{c^{(2)}_{j_2}}\ket{j_2}\big).
\end{equation}
In order to encode $c^{(1,2)}_j$ in the quantum amplitudes, we once again apply quantum adders to achieve our goals. By performing the following transformation
\begin{equation}
\ket{0}\ket{j_1}\ket{j_2} \rightarrow \ket{j}\ket{j_1}\ket{j_2},
\end{equation}
where $j \equiv (j_1+j_2) \mod N$.  This can be achieved using two quantum adders, we obtain the state
\begin{equation}
\sum_{j=0}^{N-1}\sqrt{c^{(1,2)}_j}\ket{j}\ket{\Phi_j},
\end{equation}
because the amplitude of $\ket{j}$ equals to $\sqrt{\sum_{\substack{j_1,j_2 \\ j_1+j_2 \equiv j \mod N}} \big(\sqrt{c_{j_1}c_{j_2}}\big)^2} = \sqrt{c^{(1,2)}_j}$.
\end{proof}

This algorithm can be easily extended to implementing oracles for products of $d$ circulants, in which $d$ oracles of factor circulants and $dL$ ancillary qubits are needed. Though the oracle described in \autoref{thm:product} may not be useful in all quantum algorithms, due to the additional $\ket{\Phi_j}$ in \autoref{eq:prodoracle}, it is applicable in \autoref{sec:circulant} and \autoref{sec:Hamiltonian} according to \lref{thm:sum2} (the generalized form of \lref{thm:sum}) described below. It implies that this technique could also be useful in other algorithms related to circulant matrices.

\begin{lemma}
\label{thm:sum2}
Let $M = \sum_{\alpha_j}\alpha_j W_j$ be a linear combination of unitaries $W_j$ with $\alpha_j \geq 0$ for all $j$ and $\sum_j \alpha_j = 1$. Let $O_\alpha$ be any operator that satisfies $O_\alpha \ket{0^m} = \sum_j \sqrt{\alpha_j}\ket{j}\ket{\Phi_j}$ ($m$ is the number of qubits used to represent $\ket{j}\ket{\Phi_j}$) and $\mathrm{select}(W) = \sum_j\ket{j}\bra{j}\otimes I \otimes W_j$. Then
\begin{equation}
(O_\alpha^\dagger \otimes I)\mathrm{select}(W) (O_\alpha \otimes I)\ket{0^m}\ket{\psi} = \ket{0^m}M\ket{\psi} + \ket{\Psi^{\perp}},
\end{equation}
where $(\ket{0^m}\bra{0^m} \otimes I) \ket{\Psi^\perp} = 0$.
\end{lemma}

\begin{proof}
\[
\begin{split}
(O_\alpha^\dagger \otimes I)\mathrm{select}(W) (O_\alpha \otimes I)\ket{0^m}\ket{\psi} &= (O_\alpha^\dagger \otimes I) \mathrm{select}(W) \sum_j \sqrt{\alpha_j}\ket{j}\ket{\Phi_j}\ket{\psi} \\&= (O_\alpha^\dagger \otimes I)\sum_j \sqrt{\alpha_j}\ket{j}\ket{\Phi_j} W_j\ket{\psi}
\end{split}
\]
\[
\begin{split}
(\ket{0^m}\bra{0^m}O_\alpha^\dagger \otimes I)\sum_j \sqrt{\alpha_j}\ket{j}\ket{\Phi_j} W_j\ket{\psi} &=
\ket{0^m}\sum_{j'} \sqrt{\alpha_{j'}}\bra{j'}\bra{\Phi_{j'}} \sum_j \sqrt{\alpha_j}\ket{j}\ket{\Phi_j} W_j\ket{\psi}\\
&=\ket{0^m} \sum_j \alpha_j  W_j\ket{\psi} = \ket{0^m} M \ket{\psi}
\end{split}
\]
\end{proof}

\section{Application: Solving Cyclic Systems}
\label{sec:cyclic}
Vibration analysis of mechanical structures with cyclic symmetry has been a subject of considerable studies in acoustics and mechanical engineering~\cite{Olson2014, kaveh2011block}.  Here we provide an example where the above proposed quantum scheme can outperform classical algorithms in solving the equation of motion for vibrating and rotating systems with certain cyclic symmetry.  

The equation of motion for a cyclically symmetric system consisting of $N$ identical sectors, as shown in~\autoref{fig:cyclic}, can be written as 
\begin{equation}
\label{eq:eom-gen}
M \ddot{\gv{q}} + D \dot{\gv{q}} + K \gv{q} = \gv{f},
\end{equation}
where $\gv{q}$ and $\gv{f}$ are $N$-dimensional vectors, denoting the displacement of and the external force acting on each individual sector, respectively.  The mass, damping and stiffness matrices are all circulants, represented by $M=circ(m_1, m_2, ..., m_N)$, $D=circ(d_1, d_2, ..., d_N)$ and $K=circ(s_1, s_2, ..., s_N)$. 

\begin{figure}[ht]
\centering
\subfigure[~]{\includegraphics[width=2in]{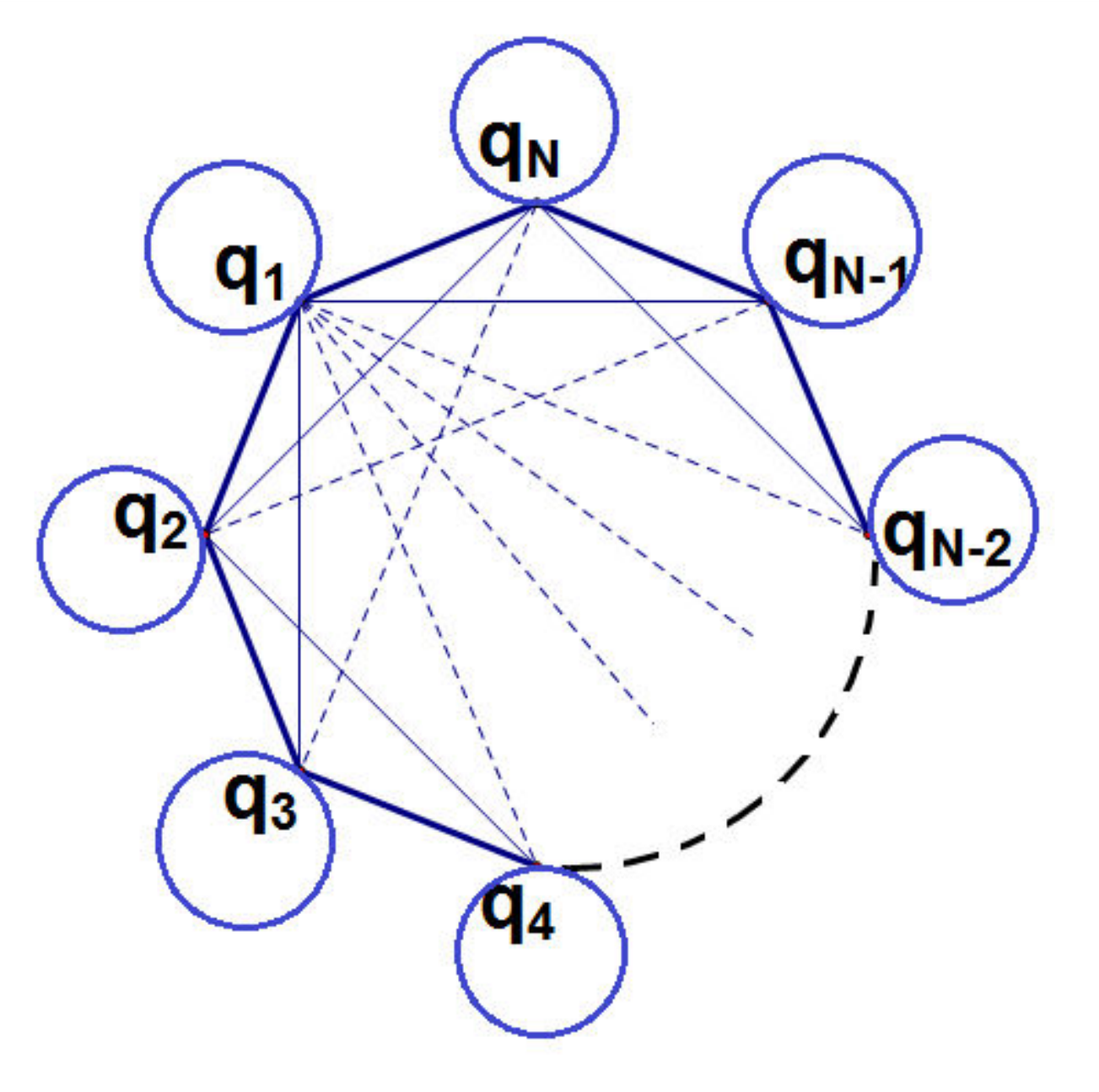}}
\hspace{0.2in}
\subfigure[~]{\includegraphics[width=2in]{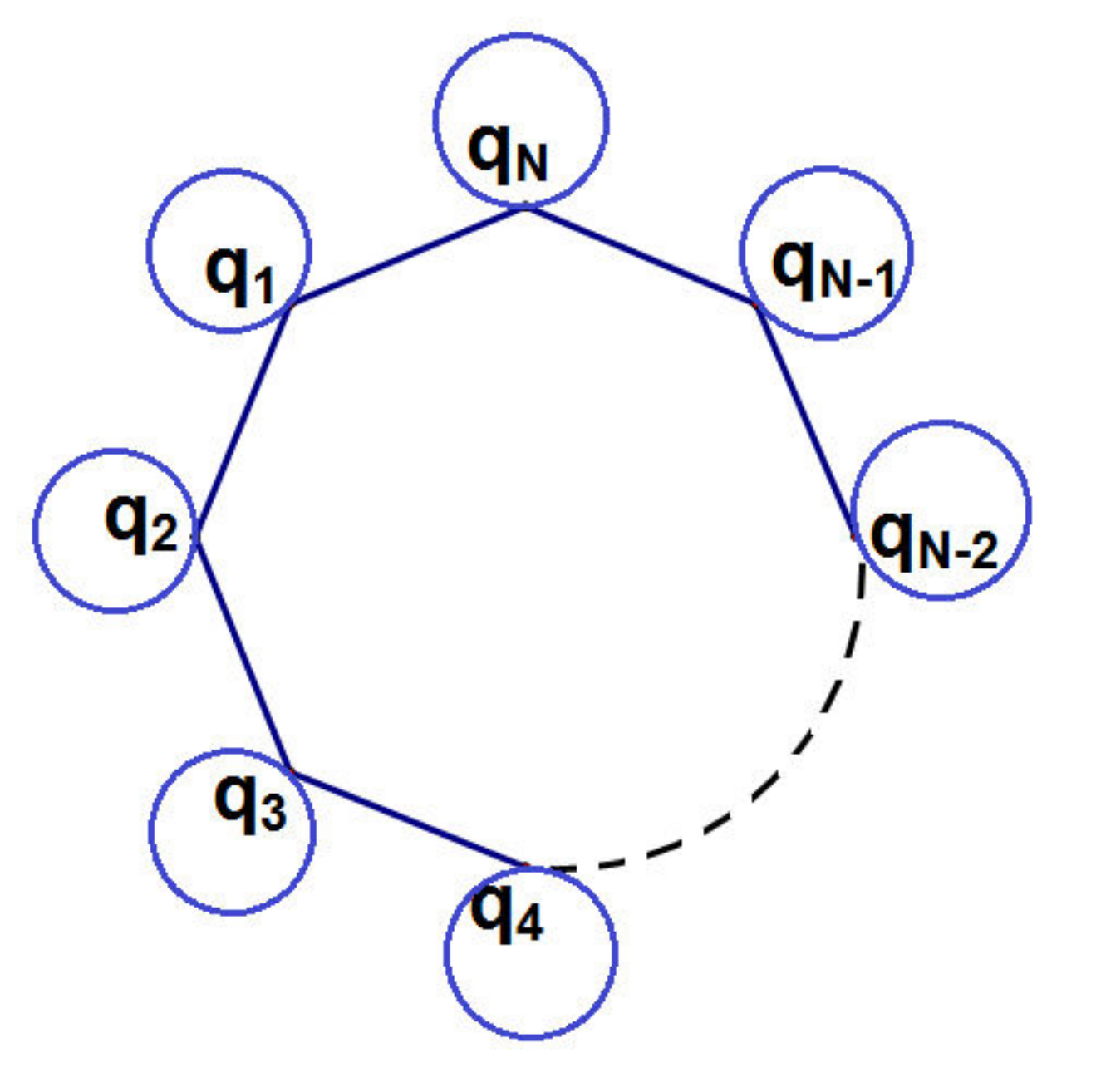}}
\caption{Topology diagram of an $N$-sector cyclic system. (a)~a general cyclic system with coupling between any two sectors which can be solved using \autoref{thm:hhl}. (b)~a cyclic system with nearest-neighbour coupling which can be solved using the HHL algorithm~\cite{harrow2009}.}
\label{fig:cyclic}
\end{figure}

Assume all sectors have the same mass ($M \propto I$) and there is zero damping ($D = 0$). If the system is under the so-called traveling wave engine order excitation, the equation of motion can be simplified as~\cite{Olson2014}: 
\begin{equation}
\label{eq:eom}
\ddot{\gv{q}} + K \gv{q} = \gv{f}e^{i n \Omega t},
\end{equation}
where the traveling wave is characterised by $f_j = f e^{i2\pi n j/N}$ for the external force vector $\gv{f}$, $n$ is the order of excitation, and $\Omega$ is the angular frequency of the excitation. We search for solutions of the form $\gv{q} = \gv{q_0} e^{i n \Omega t}$, which leads to
\begin{equation}
(K - n\Omega I)\gv{q_0} = \gv{f}.
\end{equation}
Since $K- n\Omega I$ is a circulant matrix, we can use \autoref{thm:hhl} to calculate 
$$\gv{q_0} = (K - n\Omega I)^{-1}\gv{f}.$$ 
It is important to consider the conditions under which \autoref{thm:hhl} works:
\begin{enumerate}
\item {\em $K- n\Omega I$ is Hermitian.} This is generally true for symmetric cyclic systems, where the coupling between $q_j$ and $q_{j+d}$ and the coupling between $q_j$ and $q_{j-d}$ are physically the same for any sector $j$ and distance $d$. 

\item {\em $K- n\Omega I$ has non-negative (or non-positive) entries.}  Although this is not in general true, \autoref{thm:hhl} will work under a slight modification. We observe that the off-diagonal elements of $K-n\Omega I$ are always negative because the coupling force between two connecting sectors is always in the opposite direction to their relative motion. 
\begin{itemize}
\item If the diagonal elements of $K-n\Omega I$ are also negative, then no modification to the proposed procedure is necessary. 
\item If the diagonal elements of $K-n\Omega I$ are positive, we replace $V_0$ with $-V_0$ in \autoref{eq:decomposition}, while keeping $V_j~(j\neq 0)$ unchanged. Then $-(K-n\Omega I) = - c_0V_0 + \sum_{j=1}^{N-1}c_jV_j$ would be a matrix whose off-diagonal elements are positive and diagonal elements are negative. It means that in the quantum circuits, we need to replace $\mathrm{select}(V) = \sum_{j=0}^{N-1}\ket{j}\bra{j}\otimes V_j$ with $\mathrm{Ref}_{0}\cdot\mathrm{select}(V) = -\ket{j=0}\bra{j=0}\otimes V_0 + \sum_{j=1}^{N-1}\ket{j}\bra{j}\otimes V_j$, where $\mathrm{Ref}_0 = \ket{0^L}\bra{0^L}-2I$ operating on the first register is the reflection operator around $\ket{j=0}=\ket{0^L}$.
\end{itemize}

\item {\em The condition number $\kappa$ of $K - n\Omega I$ is small.} This is true when the couplings among sectors are relatively weak --- when $\abs{K_0 - n\Omega} \gg K_1$ where $K_0$ characterises the coupling between a sector and the exterior and $K_1$ characterizes the coupling among sectors. 

\end{enumerate}
If all three conditions are satisfied, we have an exponential speed-up compared to classical computation. Note that the output $\gv{q_0}$ is stored in quantum amplitudes, which cannot be read out directly.  However, further computation steps can efficiently provide practically useful information about the system from the vector $\gv{q_0}$, for example the expectation value $\gv{q_0}^\dagger M \gv{q_0}$ for some linear operator $M$ or the similarity between two cyclic systems $\braket{q'_0|q_0}$~\cite{harrow2009}.  It is also worth noting that the proposed algorithm, in contrast to previous quantum algorithms~\cite{berry2007efficient, harrow2009, childs2010simulating, wiebe2011simulating, poulin2011quantum, berry2015simulating, mahasinghe2016efficient}, works for dense matrices $K- n\Omega I$.  It means that the cyclic systems need not be subject to nearest-neighbour coupling.

\section{Conclusion}
\label{sec:conclusion}

In this paper, we present efficient quantum algorithms for implementing circulant (as well as Toeplitz and Hankel) matrices and block circulant matrices with special structures, which are not necessarily sparse or unitary.  These matrices have practically significant applications in physics, mathematics and engineering related field.  The proposed algorithms provide exponential speed-up over classical algorithms, requiring fewer resources ($2\log N$ qubits) and having lower complexity ($O(\log^2 N/\norm{C\ket{\psi}})$) in comparison with existing quantum algorithms.  Consequently, they perform better in quantum computing and are more feasible to experimental realisation with current technology. 

Besides the implementation of circulant matrices, we discover that we can perform the HHL algorithm on circulant matrices to implement the inverse of circulant matrices, by adopting the Taylor series approach to efficiently simulate circulant Hamiltonians. Due to the special structure of circulant matrices, we prove that they are one of the types of the dense matrices that can be efficiently simulated. Being able to implement the inverse of circulant matrices opens a door to solving a variety of real-world problems, for example, solving cyclic systems in vibration analysis. Finally, we show that it is possible to construct oracles for products of circulant matrices using the oracles for their factor circulants, a technique that is useful in related algorithms.

\section*{Acknowledgments}

We would like to thank Xiaosong Ma, Anuradha Mahasinghe, Jie Pan, Thomas Loke, and Shengjun Wu for helpful discussions.

\bibliographystyle{model1-num-names}
\bibliography{circulant-refs}

\end{document}